\documentclass[aps,superscriptaddress,twocolumn,twoside,floatfix,pra,nofootinbib,a4paper]{revtex4-2}
\usepackage{enumerate,appendix}
\usepackage{amsmath, amsthm, amssymb,commath}
\usepackage{color,calc,graphicx}
\usepackage[usenames,dvipsnames,svgnames,table,cmyk,hyperref]{xcolor}
\usepackage[colorlinks]{hyperref}
\hypersetup{
	colorlinks = true,
	urlcolor = {blue},
	citecolor = {magenta},
	linkcolor= {blue}
}

\usepackage{graphicx}
\usepackage{amsmath}
\usepackage{latexsym}
\usepackage{bbm}
\usepackage{soul}

\newcommand{\ket}[1]{\vert#1\rangle}

\newcommand{\bracket}[3]{\langle#1\vert#2\vert#3\rangle}
\newcommand{\ketbra}[2]{\vert #1\rangle\langle#2 \vert}
\usepackage{physics}

\usepackage[charter,cal=cmcal,sfscaled=false]{mathdesign}
\usepackage{mathtools}
\usepackage{booktabs}
\usepackage{multirow}
\usepackage{subfigure}
\usepackage{dcolumn}
\usepackage{mathrsfs}
\usepackage{diagbox}
\usepackage{array}
\usepackage{makecell}
\usepackage{threeparttable}

\newtheorem{proposition}{Proposition}
\newtheorem{conjecture}{Conjecture}

\newtheorem{corollary}{Corollary}

\begin{document}

%%%%%%%%%%%%%%%%%%%%%%%%%%%%%%%%%%%%%%%%%%%%%%%%%%%%%%%%%%%%%%%%%%%

\title{Scalable entanglement certification via quantum communication}

\author{Pharnam Bakhshinezhad}
\email{pharnam.bakhshinezhad@tuwien.ac.at}
\thanks{P.B. previously published as Faraj Bakhshinezhad.}
\affiliation{Atominstitut, Technische Universit{\"a}t Wien, Stadionallee 2, 1020 Vienna, Austria}

\author{Mohammad Mehboudi}
\affiliation{Atominstitut, Technische Universit{\"a}t Wien, Stadionallee 2, 1020 Vienna, Austria}

\author{Carles Roch i Carceller}
\affiliation{Department of Physics and NanoLund, Lund University, Box 118, 22100 Lund, Sweden}

\author{Armin Tavakoli}
\email{armin.tavakoli@teorfys.lu.se}
\affiliation{Department of Physics and NanoLund, Lund University, Box 118, 22100 Lund, Sweden}

%%%%%%%%%%%%%%%%%%%%%%%%%%%%%%%%%%%%%%%%%%%%%%%%%%%%%%%%%%%%%%%%%%%

\begin{abstract}
Harnessing the advantages of shared entanglement for sending quantum messages often requires the implementation of complex two-particle entangled measurements.  We investigate entanglement advantages in protocols that use only the simplest two-particle measurements, namely product measurements. For experiments in which only the dimension of the message is known, we show that robust entanglement advantages are possible, but that they are fundamentally limited by Einstein-Podolsky-Rosen steering. Subsequently, we propose a natural extension of the standard scenario for these experiments and show that it circumvents this limitation. This leads us to prove entanglement advantages from every entangled two-qubit Werner state, evidence its generalisation to high-dimensional systems and establish a connection to quantum teleportation. Our results reveal the  power of product measurements for generating quantum correlations in entanglement-assisted communication and they pave the way for practical semi-device-independent entanglement certification well-beyond the constraints of Einstein-Podolsky-Rosen steering.
\end{abstract}
%%%%%%%%%%%%%%%%%%%%%%%%%%%%%%%%%%%%%%%%%%%%%%%%%%%%%%%%%%%%%%%%%%%

\maketitle

%%%%%%%%%%%%%%%%%%%%%%%%%%%%%%%%%%%%%%%%%%%%%%%%%%%%%%%%%%%%%%%%%%%

\section{Introduction}
Shared entanglement between a sender and a receiver that are connected over a quantum channel is the most powerful communication resource in quantum theory. This is famously showcased in the dense coding protocol, where entanglement doubles the classical capacity of a noise-free qubit channel  \cite{Bennett1992}. If the channel is used only once, which is the scenario most pertinent for experimental considerations, this entanglement-assisted prepare-and-measure (EAPM) scenario (see illustration in  Fig.~\ref{Fig_scenario1}) can equally well be viewed as setting for efficient quantum communication and as a platform for semi-device-independent quantum information protocols. The latter is because the state, the sender and the receiver are uncharacterised devices, and only knowledge of the dimension of the channel is required to deduce the quantum nature of the correlations. In this sense, the EAPM scenario offers an appealing path to certify the advantages of entanglement in experiments with limited characterisation.

A central obstacle for harnessing entanglement-advantages in the EAPM scenario is that protocols commonly need the receiver (Charlie) to measure both the particles, namely the one coming from the entanglement source and the one arriving over the channel, in an entangled basis. In for example optical systems, such measurements are well-known to be impossible without extra photons or nonlinear effects \cite{Lutkenhaus1999, Vaidman1999, Loock2004}, which for instance can limit experiments to using only single-photon carriers of multiple qubits (see e.g.~\cite{Huang2022, Guo2023}).  In the EAPM scenario, for the simplest case of qubit systems, a series of dense coding experiments have over time implemented increasingly sophisticated Bell basis measurements and thereby approached the theoretical limit of the entanglement advantage \cite{Mattle1996, Fang2000, Li2002, Jing2003, Schaetz2004, Schuck2006, Barreiro2008, Williams2017}. For systems of higher dimension than qubit,  the situation is extra challenging. Even resolving one element of a high-dimensional entangled basis is impossible with ancilla-free linear optics \cite{Calsamiglia2002}. The most high-dimensional optical Bell basis measurement hitherto realised is limited to three-level systems and uses ancillary photons \cite{Hu2020, Luo2019}.  In the EAPM scenario, this has led experiments based on high-dimensional entanglement and quantum communication to instead focus on simpler, suboptimal, measurements, compatible with standard linear optics \cite{Hu2018}. The challenges associated with entangled measurements are broadly relevant in the different correlation tests accommodated by the EAPM scenario \cite{Tavakoli2021, Pauwels2022, vieira2023interplays}.

\begin{figure}
	\includegraphics[width=1\columnwidth]{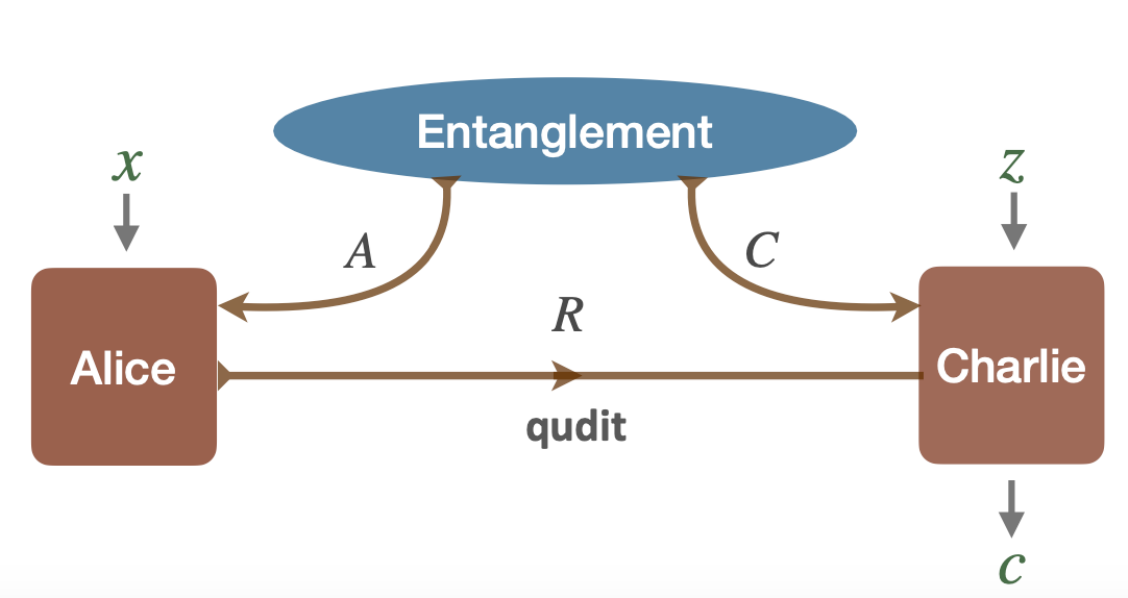}
	\caption{Entanglement-assisted prepare-and-measure scenario between a sender (Alice) and a receiver (Charlie). The information, $x$, is encoded into one share of the entangled state.}\label{Fig_scenario1}
\end{figure}

However, while entangled measurements are provably necessary for the specific task of dense coding \cite{Pauwels2022b}, this is not true in general. Interestingly, it was recently shown that there exists well-known communication tasks which in the EAPM scenario admit their best implementation in protocols that rely only on the simplest joint measurements \cite{Piveteau2022}. These are mere product measurements of the source- and message-particles; they are therefore classical post-processings of two completely separate single-particle measurements. In principle, this can greatly simplify the experiments, as the particles do not need to interfere with each other in the measurement device, and can even be measured at separate times. Nevertheless, presently, little is known about how such protocols can be constructed. Moreover, another important aspect concerns the noise-robustness of protocols based on product measurements. While entangled measurements are well-known to reveal the correlation advantages of shared entanglement even from very noisy states \cite{Abbott2018, Moreno2021}, no counterpart is known for product measurements. That is, even though product measurements can sometimes be optimal under ideal conditions, their performance might deteriorate in the presence of signficant amounts of noise in the entanglement distribution, rendering them unable to certify entanglement that is well-within the reach of schemes that use entangled measurements. Indeed, certifying noisy forms of entanglement is a central matter for correlation experiments.

Here, we investigate entanglement advantages revealed by product measurements in the EAPM scenario and show that they are much more powerful than previously known. {In all our scenarios,  the source is fully untrusted. The operations of all the parties are also untrusted, up to the bounded dimension of the quantum communication channels.} The article is structured as follows. In section~\ref{sec0} we formalise the EAPM scenario. In section~\ref{sec1}, we investigate high-dimensional entanglement by introducing concrete certification schemes in the EAPM scenario. We prove that the paradigmatic isotropic state is certified at noise rates well-above the known thresholds for Bell nonlocality, and closely resembling the thresholds known for Einstein-Podolsky-Rosen steering.  In section~\ref{sec2}, we show that the results from section~\ref{sec1} are actually close to optimal. This follows from a no-go result, in which we show that steering is a necessary condition for certifying entanglement advantages in any EAPM scenario with product measurements. In section~\ref{sec3}, we set out to circumevent this fundamental limitation. We do so by considering a natural extension of the standard EAPM scenario, which we name the symmetric EAPM scenario. In the symmetric scenario, classical information is not only encoded in one share of the state (Alice, in Fig~\ref{Fig_scenario1}) but in both shares of the state (see Fig~\ref{Fig_scenario2}). For qubit systems, we prove that every entangled Werner state implies an advantage. Finally, in section~\ref{sec4}, we introduce a prime-dimensional generalisation of the scheme in section~\ref{sec3}. We present evidence, on the basis of which we argue that every state supporting fidelity-based quantum teleportation can be certified. This notably includes every entangled isotropic state.

\section{The EAPM scenario}\label{sec0}
The EAPM scenario is illustrated in Fig~\ref{Fig_scenario1}. Alice and Charlie share a state $\rho_{AC}$, which can have an arbitrary local dimension. Alice selects a classical input $x$ and encodes it on her share of the state via a completely positive trace-preserving (CPTP) map, $\Lambda_x^{A\rightarrow R}$, whose output system ($R$), which we call the message, has a known dimension $d$. The total state arriving to Charlie becomes
\begin{equation}
\tau_x^{RC}=(\Lambda_x^{A\rightarrow R}\otimes \openone^C)[\rho_{AC}].
\end{equation}
Finally, Charlie selects a classical input $z$ and performs a joint quantum measurement $\{M^{RC}_{c|z}\}$ with outcome $c$. The Born-rule gives the quantum correlations, $p(c|x,z)=\tr\left(\tau_x^{RC} M_{c|z}\right)$. Notice that the set of states $\{\tau_x^{RC}\}$ realisable via arbitrary CPTP maps for Alice and arbitrary initial entangled state can be completely characterised as  $\tau_{x}^{RC}$ being a $d\times D$ dimensional bipartite state with $\tr_R(\tau_x^{RC})=\tau^C$, where $\tau^C$ is the reduced state which is notably independent of $x$. Note that $D$ is the dimension of the source particle, which can be arbitrary.

Our focus is on protocols where Charlie's measurements act separately on systems $R$ and $C$. This can be a product measurement followed by a classical post-processing of the respective outcomes, i.e.~
\begin{equation}\label{prod}
M^{RC}_{c|z}=\sum_{c_1,c_2} p(c|c_1,c_2) N^R_{c_1|z} \otimes N^C_{c_2|z},
\end{equation}
where $\{N^R_{c_1|z}\}$ and $\{N^C_{c_2|z}\}$ are single-system measurements and $p(c|c_1,c_2)$ is some (perhaps stochastic) rule for deciding the final outcome $c$ from the local outcomes $(c_1,c_2)$. More generally, the measurements can also be adaptive \cite{Pauwels2022b}, i.e.~Charlie could use the outcome on system $R$ to inform his measurement on system $C$, and vice versa. These adaptive product measurements take the form  $
M^{RC}_{c|z}=\sum_{c_1,c_2} p(c|c_1,c_2) N^R_{c_1|z} \otimes N^C_{c_2|z,c_1}$ and $M^{RC}_{c|z}=\sum_{c_1,c_2} p(c|c_1,c_2) N^R_{c_1|z,c_2} \otimes N^C_{c_2|z}$ respectively.

We are interested in comparing the correlations, $p(c|x,z)$, obtained from shared entanglement and product measurements, with those obtained without shared entanglement. The latters correspond to standard (entanglement-unassisted) quantum prepare-and-measure scenarios, i.e.~Alice can send any $d$-dimensional state $\alpha_x$ to Charlie who can perform an arbitrary quantum measurement on it,
\begin{equation}\label{PM}
p(c|x,z)=\tr\left(\alpha_x M_{c|z}\right).
\end{equation}
Shared classical randomness is additionally permitted between the parties.

\section{Certifying high-dimensional entanglement in the EAPM scenario}\label{sec1}
We begin by identifying a scheme in the EAPM scenario that enables us to certify entanglement under substantial and dimension-scalable noise rates. To this end, consider the following scheme. Let Alice have an input $x\equiv (x_0,x_1)\in\{0,\ldots,d-1\}^2$ and Charlie have an input $z\in\{0,\ldots,d\}$, where $d$ is prime number. The parties have the objective to compute (via the output $c$), for each $z$, a specific binary function of $x$. These functions are 
\begin{align}\nonumber
&z\neq d:  &&w_z=x_1-2zx_0 \mod{d} \\\label{wincond1}
&z=d: && w_d=x_0.
\end{align}
The average success probability of computing the functions  is therefore 
\begin{equation}\label{aspS}
\mathcal{S}_d=\frac{1}{d^2(d+1)}\sum_{x,z}p(c=w_z|x,z).
\end{equation}
Next, we will analyse $\mathcal{S}_d$ in a quantum setting with and without entanglement and prove that it certifies entanglement under product measurements.

\subsection{Protocol with shared entanglement and product measurements}
Consider now a specific quantum protocol based on shared entanglement and product measurements. Let the shared state be $\rho_{AC}=\phi^+_d$, where $\phi^+_d=|\phi^+_d \rangle\langle \phi^+_d|$ is the maximally entangled state, $|\phi^+_d\rangle=\frac{1}{\sqrt{d}}\sum_{i=0}^{d-1}\ket{ii}$. Next, define the clock and shift operators $Z=\sum_{k=0}^{d-1}e^{\frac{2\pi i k}{d}}\ketbra{k}{k}$ and $X=\sum_{k=0}^{d-1}\ketbra{k+1}{k}$, where $k+1$ is computed modulo $d$. Choose Alice's CPTP maps, $\Lambda_x$, as corresponding to the unitaries
\begin{equation}\label{WHunitary}
U_{x}=X^{x_0}Z^{x_1}.
\end{equation}
Note that for the special case of $d=2$, $X$ and $Z$ are simply two of the Pauli operators and $U_x$ is effectively the four Pauli rotations. Finally, we must select Charlie's product measurements. For the special case of $d=2$, we choose them as products of the three Pauli observables, namely
\begin{align}\label{pauli}
&E_0=X\otimes X, \qquad E_1=Z\otimes Z, \quad\text{and} \quad E_2=XZ\otimes XZ,
\end{align}
with $M_{c|z}=\frac{1}{2}(\openone+(-1)^c E_z)$. Beyond dimension two, following Eq.~\eqref{prod}, we define the measurements as post-processings of the outcomes obtained in two separate basis measurements of systems $R$ and $C$,
\begin{equation}\label{MUB}
M_{c|z}=\sum_{c_1,c_2=0}^{d-1} |e_{c_1,z}\rangle\langle e_{c_1,z}| \otimes |e_{c_2,z}^*\rangle\langle e_{c_2,z}^*| \delta_{c_1-c_2,c},
\end{equation}
where $*$ denotes complex conjugation and the addition in $\delta_{c_1-c_2,c}$ is taken modulo $d$. Notably, in odd prime dimensions, the local bases $\{|e_{m,z}\rangle\}$ are mutually unbiased. These are known to admit the form $\ket{e_{m,d}}=\ket{m}$ and
$|e_{m,z}\rangle=\frac{1}{\sqrt{d}}\sum_{l=0}^{d-1}\omega^{l(m+zl)}\ket{l}$ for $z\neq d$, where $\omega=e^{\frac{2\pi i}{d}}$ \cite{Wootters1989}. The unbiasedness property is not de-facto necessary for the success of the protocol, but is a convenient choice.

In order to evaluate the average success probability \eqref{aspS}, it is handy to first identify  the following relations, which can be straightforwardly verified,
\begin{align}\nonumber 
& X^t |e_{m,z}\rangle=\omega^{zt^2-tm} |e_{m-2zt,z}\rangle,  \\\nonumber
& Z^t |e_{m,z}\rangle= |e_{m+t,z}\rangle\\\nonumber
& X^t |e^*_{m,z}\rangle=\omega^{-zt^2+tm} |e^*_{m-2zt,z}\rangle,\\
& Z^t |e^*_{m,z}\rangle=|e^*_{m-t,z}\rangle,
\end{align}
valid for $z\neq d$ and integer $t$.  Using these, one straightforwardly finds that each of the functions is computed deterministically, that is $p(c|x,z) = \delta_{c,w_z}$, leading to $\mathcal{S}_d=1$.

\subsection{Bounding protocols without shared entanglement}
Next, we must determine a useful bound $\mathcal{S}_d\leq L_d$ valid for any quantum strategy without shared entanglement. Since this corresponds to bounding the expression \eqref{aspS} in a standard quantum prepare-and-measure scenario,  the correlations are given by Eq.~\eqref{PM}. The relevant quantity becomes
\begin{equation}\label{quantityeq}
\max_{\{\alpha_x\}, \{M_{c|z}\}} \frac{1}{d^2(d+1)}\sum_{x,z} \tr\left(\alpha_x M_{w_z|z}\right),
\end{equation}
where $\alpha_x$ is a $d$-dimensional state. We restrict the analysis to prime number dimensions because in these cases the conditions in Eq.~\eqref{wincond1} are particularly hard to meet without entanglement. The task at hand can be seen as an (unorthodox) variant of a quantum random access code \cite{Ambainis1999, Tavakoli2015}. The proof ideas recently developed for quantum random access codes in Ref.~\cite{farkas2023simple} can be immediately modified to obtain a general bound, $L_d$, on Eq.~\eqref{quantityeq}, namely
\begin{equation}\label{Ld}
L_d=\frac{1}{d}\left(1+\frac{d-1}{\sqrt{d+1}}\right),
\end{equation}
for prime $d$. The derivation is detailed in  Appendix~\ref{App:firstgame} and it is based on analysing operator norms for sums of the measurement operators. Regardless of the protocol used, the observation of $\mathcal{S}_d>L_d$ implies the certification of entanglement. The bound $L_d$ is typically not tight (except for $d=2$), i.e. it does not equal the value defined in \eqref{quantityeq}. The reason for this becomes apparent in Appendix~\ref{App:firstgame}, where both operator norm inequalities and concavity inequalities are employed, the saturation of which is not guaranteed in general. Nevertheless, to give an indication of how close to optimal the bound is, we have numerically optimised the argument in \eqref{quantityeq} over the set of quantum states and measurements. For $d=3,5,7$ we obtain the lower bounds $0.6616$, $0.5121$ and $0.4233$ respectively, which can be compared to the upper bounds $0.6667$, $0.5266$ and $0.4459$ obtained respectively from \eqref{Ld}.  We note that numerical techniques likely can be used to improve the bound \eqref{Ld}, but only for specific values of $d$ \cite{Rosset2019}.

Even though the bound \eqref{Ld} is not generally tight, it is good enough to reveal the qualitative abilities of product measurements in a dimension-scalable manner. To showcase that, we focus on the seminal isotropic state,
\begin{equation}
\rho_v^\text{iso}=v\phi^+_d+\frac{1-v}{d^2}\openone,
\end{equation}
with visibility $v\in[0,1]$. Thus, when running the strategy from the previous section, we compute the smallest visibility for which the state produces a value of $\mathcal{S}_d$ that exceeds the limit in Eq.~\eqref{Ld}. For comparison, the isotropic state is known to be entangled if and only if $v>\frac{1}{d+1}$ \cite{Horodecki1996}.

\begin{corollary}\label{prop1}
For every prime dimension $d$,  entanglement certification in the EAPM scenario with product measurements is possible for the isotropic state when
\begin{equation}\label{vlim}
v>\frac{1}{\sqrt{d+1}}.
\end{equation}
\end{corollary}
This exhibits an inverse-square-root scaling in the dimension parameter, thus showing that product measurements become increasingly good at certifying the entanglement. In particular, for $d=2$, it reduces to $v>1/\sqrt{3}$, which significantly improves on previous protocols \cite{Piveteau2022}, and happens to equal the exact threshold for steerability of $\rho_v^\text{iso}$ for the same number (three) of measurements \cite{Bavaresco2017}. Moreover, for prime $d$, Eq.~\eqref{vlim} exactly matches the bound for steerability of $\rho_v^\text{iso}$ under $d+1$ mutually unbiased bases obtained from the steering inequality of Ref.~\cite{Marciniak2015}.

\section{No entanglement advantage without steering}\label{sec2}
%\moha{2---I think as a result a similar statement can be made "without incompatible measurements" but who cares. Suppose instead we fix the state to maximally entangled. Measurements are to be decided. What is the \textit{degree} of measurement incompatibility---according to what measure?---required to observe quantum improvement?}
It is not a coincidence that our above scheme happens to give critical visibilities that closely parallel results known for steering. As we now show, the above results are examples of saturation (or near saturation) of a more fundamental limitation that applies to any protocol in the EAPM scenario using adaptive product measurements. 

\begin{proposition}\label{prop2}
Let  $\rho_{AC}$ be any entangled state that is not steerable from $C$ to $A$. Then, any probability distribution in the EAPM scenario obtained from adaptive product measurements can be simulated in a quantum model with shared classical randomness.
\end{proposition}
\begin{proof}
Consider first product measurements that are adaptive from system $C$ to system $R$. The probability distribution can then be written 
\begin{equation}
p(c|x,z)=\sum_{c_1,c_2}p(c|c_1,c_2)\Tr\left(\Lambda_x[\rho_{c_2|z}]M^{R}_{c_1|z,c_2}\right),
\end{equation}
where $\rho_{c_2|z}=\Tr_C\left(\openone\otimes M^{C}_{c_2|z}\rho_{AC}\right)$ are the unnormalised states remotely prepared on $A$ by measuring $C$. If $\rho_{AC}$ is unsteerable from $C$ to $A$, there exists a  local hidden state decomposition $\rho_{c_2|z}=\sum_\lambda p(\lambda)p(c_2|z,\lambda)\tau_\lambda$ for some arbitrary-dimensional quantum states $\{\tau_\lambda\}$.  Inserting this, we obtain
\begin{equation}\label{qq4}
p(c|x,z)=\sum_{\lambda,c_1,c_2} p(\lambda)p(c|c_1,c_2) p(c_2|z,\lambda) \Tr\left(\Lambda_x[\tau_\lambda]M^{R}_{c_1|z,c_2}\right).
\end{equation}	
This can be simulated without entanglement as follows. Let Alice and Charlie share $\lambda$, with distribution $p(\lambda)$. Alice prepares $\tau_\lambda$ and  applies  $\Lambda_x$ to it, sending the  $d$-dimensional state $\Lambda_x[\tau_\lambda]$ to Charlie. He draws  $c_2$ from the distribution $p(c_2|z,\lambda)$, then  applies the measurement $\{M^{R}_{c_1|z,c_2}\}$, and lastly draws $c$ from $p(c|c_1,c_2)$.  This reproduces the distribution in Eq.~\eqref{qq4}.

The case of product measurements adaptive from system $R$ to system $C$ is similarly treated. The probability distribution becomes
\begin{equation}
p(c|x,z)=\sum_{c_1,c_2}p(c|c_1,c_2)\Tr\left(\Lambda_x[\rho_{c_2|z,c_1}]M^{R}_{c_1|z}\right),
\end{equation}
where $\rho_{c_2|z,c_1}=\Tr_C\left(\openone\otimes M^{C}_{c_2|z,c_1}\rho_{AC}\right)$ are the unnormalised states remotely prepared on $A$. The existence of a local hidden state model implies 
\begin{equation}
p(c|x,z)=\sum_{\lambda,c_1,c_2} p(\lambda)p(c|c_1,c_2) p(c_2|z,c_1,\lambda) \Tr\left(\Lambda_x[\tau_\lambda]M^{R}_{c_1|z}\right).
\end{equation}	
To simulate this distribution without entanglement, one distributes $\lambda$, let's Alice prepare $\tau_\lambda$ and run it through the map $\Lambda_x$. Charlie first measures the message, then uses the outcome $c_1$ to draw $c_2$ from $p(c_2|z,c_1,\lambda)$ and finally draws $c$ from $p(c|c_1,c_2)$.
\end{proof} 

A noteworthy corollary of this argument is that,  product measurements adaptive from $C$ to $R$, in an EAPM scenario with $N$ inputs for Charlie, one-way steerability under just $N$ measurements is necessary for an entanglement-advantage. This makes a significant difference since steerability under a limited number of measurements is known to be considerably more constrained than steerability under unboundedly many measurements \cite{Bavaresco2017}.  In view of this, the scheme from the previous section, which led to Corollary~\ref{prop1} via independent product measurements, is optimal for $d=2$ since it coincides with the steering bound of $\rho_v^\text{iso}$ under three measurements.  For larger $d$, it is unlikely that our result from the previous section can be much improved, because of the steering results for $N=d+1$ bases in \cite{Marciniak2015, Designolle2019}. However, by employing potentially unboundedly many measurements (instead of $d+1$ as in our case), it may be possible to approach the ultimate steering limit  \cite{Wiseman2007} on the parameter $v$.

Proposition~\ref{prop2} provides a fundamental limitation on the abilities of product measurements in the EAPM scenario. Although we already found that significantly noise-tolerant entanglement certification is possible, it is impossible to certify any state which is entangled but not steerable. Therefore, in what follows, we go beyond the EAPM scenario to show that this obstacle can be overcome, allowing for even stronger entanglement certification.

\section{The symmetric EAPM scenario}\label{sec3}
In order to circumvent the limitation on product measurement schemes imposed by Proposition~\ref{prop2}, we consider an extended version \cite{Abbott2018, Piveteau2023} of the original EAPM scenario which we refer to as the symmetric EAPM scenario. The extension is modest in terms of an implementation perspective and is conceptually natural. In the original EAPM scenario, classical information is only encoded into half the entangled state, namely by Alice, into system $A$. In the symmetric EAPM scenario, we want to encode classical information also in the other half of the entangled state. To make this possible, we add a third party, Bob, who selects an input $y$ and encodes it into the second source particle before relaying it to Charlie. See illustration in Fig~\ref{Fig_scenario2}.

Let us now write the state as $\rho_{AB}$, distributed to Alice and Bob. They each select $x$ and $y$ and perform CPTP maps $\Lambda_x^{A\rightarrow R_1}$ and $\Gamma_y^{B\rightarrow R_2}$, with the output systems $R_1$ and $R_2$ each being of dimension $d$. These are now separate, but potentially entangled, messages. The measurements of Charlie, $\{M_{c|z}^{R_1R_2}\}$, are applied jointly to both messages, leading to the quantum correlations 
\begin{equation}\label{entcorr}
p(c|x,y,z)=\tr\left((\Lambda^{A\rightarrow R_1}_x\otimes\Gamma^{B\rightarrow R_2}_y)[\rho_{AB}] M^{R_1R_2}_{c|z}\right).
\end{equation}
We remark that all parties can also share classical randomness, which can be included in the state $\rho_{AB}$.

\begin{figure}
	\centering
	\includegraphics[width=\columnwidth]{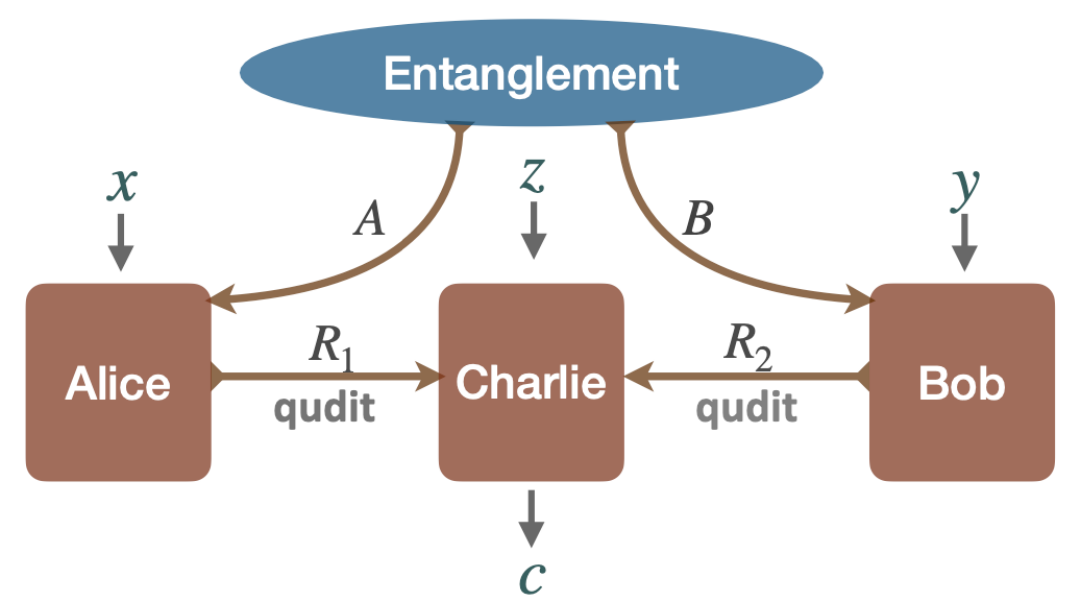}
	\caption{Symmetric entanglement-assisted prepare-and-measure scenario between the senders (Alice, Bob) and the receiver (Charlie). The information, ($x$,$y$), is encoded into the shares of the entangled state.}\label{Fig_scenario2}
\end{figure}

Again, we are interested in how protocols using shared entanglement and product measurements can outperform protocols using no shared entanglement. The correlations from the latters are described as
\begin{equation}\label{symPM}
p(c|x,y,z)=\tr\left(\alpha_x\otimes \beta_y M_{c|z}\right), 
\end{equation}
where $\alpha_x$ and $\beta_y$ are $d$-dimensional states sent from Alice and Bob, respectively, to Charlie. Note that in contrast to the original EAPM scenario, the measurement $\{M_{c|z}\}$ can now be entangled.

A direct inspection shows that the argument used to arrive at Proposition~\ref{prop2} cannot be repeated for the symmetric EAPM scenario. Indeed, the argument relies on the fact that Alice's operations do not influence the other particle, making one-way steering relevant. The counterparts to these states arriving to Charlie are now influenced by Bob. We shall see that this is not a superficial observation; the symmetric EAPM scenario can indeed certify unsteerable entanglement. To show this, let us focus on qubit systems and the following scheme.

Alice and Bob each select two bits, $x\in(x_0,x_1)\in\{0,1\}^2$ and $y\in(y_0,y_1)\in\{0,1\}^2$. Charlie selects one of three inputs $z\in\{0,1,2\}$, each with binary outputs $c\in\{0,1\}$. The aim is for Charlie to compute a binary function for each $z$, specifically the functions 
\begin{align}\nonumber
&z=0: \qquad w_0=x_0+ y_0 \\\nonumber
&z=1: \qquad w_1=x_1+ y_1 \\\label{wincond2}
&z=2: \qquad w_2=x_0+ x_1+ y_0+ y_1,
\end{align}
computed modulo $2$. The average success probability becomes
\begin{equation}\label{R2}
\mathcal{R}_2=\frac{1}{48}\sum_{x,y,z}p(c=w_z|x,y,z).
\end{equation}

In analogy with the discussion in section~\ref{sec1}, this task can be performed deterministically with shared entanglement and product measurements. In complete analogy with the protocol in section~\ref{sec1}, we let $\rho_{AB}=\phi^+_2$ and we let Charlie perform the separate Pauli observables in Eq.~\eqref{pauli}. Alice and Bob both select among the same four Pauli unitaries, namely $U_x$ and $U_y$, as given in Eq.~\eqref{WHunitary}. Evaluating Eq.~\eqref{R2} then gives $\mathcal{R}_2=1$.

The key question is to determine the largest value of $\mathcal{R}_2$ achievable in a quantum model without shared entanglement.  Consider first a classical protocol, in which $\alpha_x$ and $\beta_y$ in Eq.~\eqref{symPM} are all diagonal in the same basis. An optimal strategy is for Alice and Bob to simply send $x_0$ and $y_0$ respectively, leading to a deterministic output for $z=0$ but random outputs for $z\in\{1,2\}$, and thus a value of $\mathcal{R}_2=\frac{2}{3}$. 
We prove in Appendix~\ref{App:qubitproof} that this cannot be improved in a generic quantum protocol without shared entanglement, i.e.~any model of the form \eqref{symPM} obeys $\mathcal{R}_2\leq \frac{2}{3}$. Consequently, any quantum-over-classical advantage in the scheme must be due to entanglement. Notably, the same is not true for the scheme presented in section~\ref{sec1}; there the classical limit can be overcome using quantum communication without entanglement, and then be further enhanced by adding entanglement. %Therefore, entanglement is not responsible for some of the non-classicality.\moha{I don't follow what this sentence means. There, enhancement was due to steerability, no?}

We remark that the proof presented in Appendix~\ref{App:qubitproof} applies more generally. It can be used to bound the average success probability in any input-output scenario in which  Charlie has binary outputs and the winning conditions are XOR between balanced functions of Alice's input and Bob's input. Nevertheless, we focus on the specific case in Eq.~\eqref{wincond2} because of its relevance for certifying the entanglement of isotropic states\footnote{Note that the qubit isotropic state is up to local unitaries equivalent to the qubit Werner state.}. Indeed, a simple calculation now shows that every entangled isotropic state is certified, i.e.~$\mathcal{R}_2>\frac{2}{3}$ when $v>\frac{1}{3}$. In contrast, the state is steerable under generic projective measurements only when $v>\frac{1}{2}$ \cite{Wiseman2007}. Notably, this result completely solves the main open problem raised in Ref.~\cite{Piveteau2022}. 

More generally,  consider the so-called maximally entangled fraction of $\rho_{AC}$,
\begin{equation}
\text{EF}_d(\rho)=\max_{\Lambda_1,\Lambda_2} \bracket{\phi^+_d}{(\Lambda_1 \otimes \Lambda_2) [\rho]}{\phi^+_d}
\end{equation}
where $\Lambda_1$ and $\Lambda_2$ are CPTP maps with $d$-dimensional output spaces. A non-trivial maximally entangled fraction corresponds to $\text{EF}_d(\rho)> \frac{1}{d}$, and it is the key parameter for quantifying fidelity-based quantum teleportation \cite{Horodecki1999}. We find that it gives a sufficient condition for whether a state can be certified via product measurements in our scheme.
\begin{proposition}\label{prop3}
Every state  $\rho_{AB}$ with a non-trivial qubit maximally entangled fraction can be certified in the symmetric EAPM scenario using product measurements. In particular, it can achieve the value 
	\begin{equation}\label{qubitSF}
\mathcal{R}_2=\frac{1}{3}+\frac{2}{3}{\rm EF}_2(\rho_{AB}).
\end{equation}	
Moreover, this value is optimal when $\rho_{AB}$ is a pure two-qubit state.
\end{proposition}

\begin{proof} Here, we show only that Eq.~\eqref{qubitSF} is attainable, with remaining details given in Appendix~\ref{App:SFqubit}.  Upon receiving the shares of $\rho_{AC}$, let Alice and Bob first apply some arbitrary CPTP maps $\Lambda_1$ and $\Lambda_2$ respectively, whose output systems are $d$-dimensional. Subsequently, they each implement the previously given protocol, i.e.~they perform unitaries $U_x$ and $U_y$ respectively and Charlie measures the observables in Eq.~\eqref{pauli}. We can express the average success probability as  
	\begin{equation}\label{stepres}
	\mathcal{R}_2=\frac{1}{2}+\frac{1}{96} \tr\left((\Lambda_1 \otimes \Lambda_2) [\rho]  \sum_z \mathcal{B}^{(1)}_z \otimes \mathcal{B}^{(2)}_z   \right),
	\end{equation}
	where
	\begin{align}\nonumber
	&\mathcal{B}_0^{(1)}=\sum_{x}(-1)^{x_0}U_{x}^\dagger X U_{x}=4X \\\nonumber
	& \mathcal{B}_1^{(1)}=\sum_{x}(-1)^{x_1}U_{x}^\dagger Z U_{x} =4Z\\
	& \mathcal{B}_2^{(1)}=\sum_{x}(-1)^{x_0+x_1}U_{x}^\dagger XZ U_{x}=4XZ
	\end{align}
	 The right-hand-sides are obtained after some simplifications.  Due to the protocol's symmetry, we have $\mathcal{B}_z^{(1)}=\mathcal{B}_z^{(2)}$. One then observes that, 
	\begin{equation}
	\sum_z \mathcal{B}^{(1)}_z \otimes \mathcal{B}^{(2)}_z	=16(4\phi^+_2 -\openone).
	\end{equation}
	Inserted into Eq.~\eqref{stepres} and allowing for an optimisation over the channels $\Lambda_1$ and $\Lambda_2$, we obtain Eq.~\eqref{qubitSF}.
\end{proof}
Thus, a state's usefulness in teleportation is a sufficient condition for certification in the symmetric EAPM scenario. Notably,  many states with a non-trivial maximally entangled fraction do not admit any steering. The most  immediate example is  the isotropic state in the interval $\frac{1}{3}<v<\frac{1}{2}$ \cite{Wiseman2007}. %We remark that for pure two-qubit states, no larger value than \eqref{qubitSF} possible, as shown in Appendix~\ref{App:SFqubit}.
We remark that we have numerically explored the trade-off between $\mathcal{R}_2$ and the set of (mixed) two-qubit states with a bounded maximally entangled fraction, and we again find that Eq.~\eqref{qubitSF} is the optimal value of $\mathcal{R}_2$ for every such state.

\section{Towards high-dimensional schemes}\label{sec4}
Having found that the symmetric EAPM scenario can for some classes of states enable even optimal entanglement advantages under product measurements, we proceed with investigating whether the same is possible also for high-dimensional systems. To this end, we draw inspiration from the scheme in section~\ref{sec1} and extend it to the symmetric EAPM scenario. 

Let $d$ be an odd prime number. Let Alice and Bob each select one of $d^2$ inputs, $x=(x_0,x_1)\in\{0,\ldots,d-1\}^2$ and $y=(y_0,y_1)\in\{0,\ldots,d-1\}^2$. Charlie selects $z\in\{0,\ldots,d\}$ and outputs $c\in\{0,\ldots,d-1\}$. The winning conditions correspond to computing the following functions
\begin{align}\nonumber
&z\neq d:  &&w_z=x_1+y_1-2z(x_0-y_0) \mod{d} \\\label{wincond3}
&z=d: && w_d=x_0-y_0\mod{d}.
\end{align}
The average success probability of computing these functions is
\begin{equation}
\mathcal{R}_d=\frac{1}{d^4(d+1)}\sum_{x,y,z}p(c=w_z|x,y,z). \label{RdsymEAPM}
\end{equation}
Notice that for $d=2$, this reduces to the qubit scheme from section~\ref{sec3}.

A protocol based on product measurements, analogous to that used in section~\ref{sec1}, can deterministically compute each of the winning functions. That is, choose $\rho_{AB}=\phi^+_d$, choose Alice's and Bob's unitaries as in Eq.~\eqref{WHunitary} and choose Charlie's measurements as in Eq.~\eqref{MUB}, with the $d+1$ mutually unbiased bases $\{\vert e_{m,z}\rangle\}$. One then calculates that $\mathcal{R}_d=1$.

Consider now a fully classical model. A simple protocol is, just like for $\mathcal{R}_2$, to send e.g.~$x_0$ and $y_0$ to Charlie and thus let him output correctly ($c=w_z$) when $z=d$ but output at random when $z\neq d$. This leads to $\mathcal{R}_d=2/(d+1)$. One may wonder weather there exist a quantum strategy without entanglement that improves this bound.
Nonetheless, in analogy with what was proven for the qubit case in section~\ref{sec3}, we are unable to find any such protocol.  
Particularly, when employing the strategy mentioned above, that was optimal for shared entanglement, but now to  the case without shared entanglement, we get the classical score ${\cal R}_d = 2/(d+1)$---see Appendix~\ref{App:SFqudit} for more details.
While for $d=2$ we proved analytically that the bound cannot be improved, in Appendix~\ref{App:qubitproof}, for the cases of $d=3$ and $d=5$, we have used a numerical search based on the see-saw method \cite{tavakoli2023semidefinite} to optimise $\mathcal{R}_d$ over the operations of Alice, Bob and Charlie without shared entanglement. Specifically, we optimise $\mathcal{R}_d$ in Eq.~(\ref{RdsymEAPM}) for all possible correlations according to Eq.~(\ref{symPM}) for any set of local quantum states $\alpha_x$, $\beta_y$ in Alice and Bob’s laboratories respectively, and measurements $M_{c|z}$ in Charlie’s laboratory. The optimization is rendered as a semidefinite program with variables iterating in a see-saw manner. That is, we begin sampling random quantum states $\alpha_x$ and $\beta_y$ with dimension $d$ and optimise $\mathcal{R}_d$ for any measurements of Charlies, $M_{c|z}$. The optimal $M_{c|z}$ are stored and $\mathcal{R}_d$ is again optimised but now over all possible states of Alice, $\alpha_x$. Again, the optimal $\alpha_x$ are stored and now the optimisation runs over all possible states of Bob, $\beta_y$. This routine of three separate optimisations is then repeated until the estimated value of $\mathcal{R}_d$ converges (within a precision factor of  $10^{-4}$). In over $300$ separate trials for each $d$, we have without exception found the obtained value of $\mathcal{R}_d$, coincides with the classical bound. On this basis, we make the following conjecture.

\begin{conjecture}\label{Conj}
	For every odd prime $d$, the largest average success probability achievable in a quantum model without shared entanglement is
	\begin{equation}\label{conj}
	\mathcal{R}_d =\frac{2}{d+1}.
	\end{equation}
\end{conjecture}

Interestingly, if the conjecture is true, it implies that the strong entanglement advantages previously proven for qubit systems can be extended to high-dimensional systems. In Appendix~\ref{App:SFqudit}, we prove that the connection between the maximally entangled fraction and the average success probability, seen in Proposition~\ref{prop3}, generalises to larger $d$.
\begin{proposition}\label{prop4}
	For every odd prime $d$ and state $\rho_{AB}$, there exists a quantum model achieving the average success probability
	\begin{equation}\label{quditSF}
	\mathcal{R}_d=\frac{1}{d+1}+\frac{d}{d+1}\text{EF}_d(\rho).
	\end{equation}
\end{proposition}
Moreover, numerics for $d=3$ suggests that for pure states the value \eqref{quditSF} is optimal, but reveals that the same is not always true  for mixed states.  If Conjecture~\ref{Conj} is true, Proposition~\ref{prop4} implies that every state with a non-trivial maximally entangled fraction exceeds the limitation \eqref{conj} and is therefore certified as entangled. In particular, every entangled isotropic state $\rho_v^\text{iso}$ has a non-trivial maximally entangled fraction and therefore this family of states is optimally certified.

\section{Discussion}

\begin{table*}[]
	\begin{tabular}{|c|c|c|c|}
		\hline
		& Qubit                 & Qudit                 & \begin{tabular}[c]{@{}c@{}}Fundamental \\ limitation\end{tabular} \\ \hline
		EAPM                                                         &                $v>\frac{1}{\sqrt{3}}$       &           $v>\frac{1}{\sqrt{d+1}}$            & One-way steering                                                  \\ \hline
		\begin{tabular}[c]{@{}c@{}}Symmetric \\ EAPM\end{tabular}    &             $v>\frac{1}{3}$          &      \cellcolor[HTML]{7CFFA8}$v>\frac{1}{d+1}$                 & ?                                                                 \\ \hline
		\begin{tabular}[c]{@{}c@{}}Steering \\ (via MUBs)\end{tabular} &        $v>\frac{1}{3}$             &        $v>\frac{1}{\sqrt{d+1}}$   \cite{Marciniak2015, Skrzypczyk2015}             &    \begin{tabular}[c]{@{}c@{}}Number \\ of MUBs \end{tabular}                                                              \\ \hline
		\begin{tabular}[c]{@{}c@{}}Steering (general)\end{tabular}   &         $v>\frac{1}{2}$              &       $v>\frac{H_d-1}{d-1}$   \cite{Wiseman2007}             & \begin{tabular}[c]{@{}c@{}}Local hidden \\ states\end{tabular}    \\ \hline
		Dense coding                                                 &         $v>\frac{1}{3}$              &      $v>\frac{1}{d+1}$                  & \begin{tabular}[c]{@{}c@{}}Bell state \\ measurement\end{tabular} \\ \hline
		\multicolumn{1}{|l|}{Bell nonlocality}                       & $v>0.6961$ \cite{Designolle2023} &\begin{tabular}[c]{@{}c@{}} $v>0.6734$ \cite{Collins2002}\\  for $d\rightarrow \infty$\end{tabular} & \begin{tabular}[c]{@{}c@{}}Local hidden\\  variables\end{tabular} \\ \hline
	\end{tabular}
	\caption{Summary of results and comparison with other approaches. The \textit{Qubit} and \textit{Qudit} columns indicate bounds on the critical visibility for certifying the isotropic state. In the case of the symmetric EAPM scenario (also dense coding and general steering), the results are optimal. In the case of \textit{Qubit Bell nonlocality} the result is known to be nearly optimal \cite{Designolle2023} but for \textit{Qudit Bell nonlocality} the optimal bound is largely an open problem. We have included a comparison with the best known steering bound for protocols based on complete sets of mutually unbiased bases since our construction in the EAPM scenario also is based on these bases. The coloured box assumes Conjecture~\ref{Conj}. Proposition~\ref{prop2} shows a fundamental limitation in the EAPM scenario. Whether a corresponding limitation exists in the symmetric EAPM scenario is an open problem but it must be an entanglement concept that is weaker than usefulness in  fidelity-based quantum teleportation.}\label{Tab}
\end{table*}

We have shown that product measurements are sufficient for revealing the advantages of  noisy forms of entanglement in prepare-and-measure scenarios, and that this can be achieved via simple communication tasks. In the standard EAPM scenario, we showed that visibility requirements for white noise  can decrease as the inverse square-root of the dimension.  However, we also showed that this scalability is fundamentally limited by a need for steerability for entanglement advantages. By proposing the symmetric EAPM scenario, we showed how this limitation can be overcome, sometimes even in an optimal way. This is exemplified by every entangled qubit Werner state being certified, as well as every state useful for fidelity-based teleportation. Beyond qubits in the symmetric EAPM scenario, we also showed how these results can be generalised to prime-dimensional systems, but this ultimately requires a proof of  Conjecture~\ref{Conj}. Extending our methods to arbitrary, non-prime, dimensions is a natural next step. 

Our results pave the way for theoretical exploration and experimental implementation of strong forms of  semi-device-independent entanglement certification, which may apply also to finer entanglement concepts such as Schmidt numbers or fidelity estimation with a target state.  It is appropriate to label our scenarios as semi-device-independent, because they require none of the quantum devices to be perfectly characterised, but \textit{only} that the number of degrees of freedom in the channel is bounded. Therefore, this form of entanglement certification is far stronger than standard entanglement witnessing, where the devices are assumed to be flawless. For instance, the latter is known to be vulnerable to false positives when devices do not precisely correspond to the desired measurement \cite{Rosset2012, Morelli2023, cao2023genuine}.

The main results of this work are summarised in the first two rows of Table~\ref{Tab}, and the rest of the table compares our results with other relevant types of protocols. The table focuses on the well-known isotropic state for the sakes of simplicity and providing a concrete benchmark for the protocols. However, in general our results apply to arbitrary states, as no assumption on the entanglement source is required. Using Table~\ref{Tab}, we now proceed to discuss our results in this broader context of entanglement certification.

Firstly, Table~\ref{Tab} shows that our protocol for the EAPM scenario, which is based on measuring products of complete sets of MUBs, has the same certification performance as steering protocols based on complete sets of MUBs \cite{Marciniak2015, Skrzypczyk2015}; at least when one uses the best known closed expression for the performance of the two protocols. However, the exact performance of both protocols is underestimated since a precise analytical solution is not known in both cases. Notably, if one considers general steering protocols, with infinitely many measurements, the critical visibility can be further reduced \cite{Wiseman2007}. It is an interesting conceptual question whether there exists product measurement protocols in the EAPM scenario that, in the limit of many measurements, can reach the critical  visibility for steering, which is $v=(H_d-1)/(d-1)$, where $H_d=\sum_{k=1}^d 1/k$. However, our protocols in the symmetric EAPM scenario, again using  products of complete sets of MUBs, outperform significantly the general steering bound. Due to our use of product measurements, we achieve this while using similar experimental resources as employed in steering experiments. Interestingly, from the point of view of the assumptions made on the system, we require only a dimension bound on the channel, which is often less severe than the assumption that one measurement device is flawlessly characterised, which is employed in steering. Notably, a strict dimension assumption can also be relaxed so that undesired  small high-dimensional components associated with the implementation can be taken into account \cite{AlmostQudit}. 
	
Secondly, it is well-known that $d$-dimensional dense coding protocols can detect every isotropic entangled state, namely $v>\frac{1}{d+1}$ \cite{Abbott2018}. As we proved for qubits and conjectured for  higher dimensions, the same holds for our protocol in the symmetric EAPM scenario. In this sense, we preserve the certification power for the isotropic state while greatly reducing experimental requirements; from deterministic and complete Bell state measurements to product measurements of separate systems. It is relevant to note that we are not the first to realise that product measurements can give rise to strong quantum correlations in the EAPM scenario, as this was reported in Ref.~\cite{Piveteau2022}. However, the protocol proposed there works only for qubits, and while it is optimally implemented with product measurements, it has a very small noise tolerance. Specifically, it achieves $v=1/\sqrt{2}\approx 0.7071$, compared to our $v=1/\sqrt{3}\approx 0.5774$ in the EAPM scenario and $v=1/3$ in the symmetric EAPM scenario. Note that the assumptions in these protocols are always the same, namely a dimension bound on the channel.

Thirdly, we can compare our certification results to those obtained in Bell inequality tests. Certification via nonlocality is conceptually the strongest, since it requires no assumptions beyond the validity of quantum theory, but the certification performance is more limited. Little is known about the possibility of violating Bell inequalities with isotropic states beyond dimension two. To our knowledge, the best bound on $v$ is that reported in \cite{Collins2002}; it  decreases monotonically with $d$ but converges only to $v=0.6734$ in the limit of large $d$. A more certain comparison is possible in the qubit case; here the  optimal known visibility is $v\approx 0.6961$ and is known that no Bell inequality can reduce it below   $v\approx 0.6875$ \cite{Designolle2023}. In contrast, our protocols in both the symmetric and standard EAPM scenarios achieve certification at signficantly smaller visibilities.
	
Moreover, as noted in Table~\ref{Tab}, it is possible that  product measurement protocols in the symmetric EAPM scenario are fundamentally limited by some operational notion of nonclassicality that is  weaker than one-way steering but stronger than entanglement. Given our results, one may be inclined to suggest that the relevant concept is usefulness in fidelity-based teleportation. However, this is not accurate because we can numerically find entanglement advantages from states with a trivial maximally entangled fraction.

{Furthermore, in our protocols, Charlie always performs product measurements. However, sometimes it can be practically costly to communicate the quantum messages from Alice and Bob to Charlie. We note that this can be circumvented by ``splitting'' Charlie into two separate parties, one neighbouring Alice and one neighbouring Bob, with independent inputs $z$ and $z'$. By associating these inputs to the respective single-particle measurements entering Charlie's product measurement, we can recover the same statistics as in our protocols by imposing the post-selection condition $z=z'$.}

Another relevant discussion is that of closing the detection loophole. Our protocols were not designed with the aim of minimising detection requirements, but they nevertheless perform well in this regard.  Deterministic and complete entangled measurements on separate photons are well-known to be complicated and require additional resources such as nonlinear optics or auxiliary qubits, see e.g.~\cite{Kim2001, Barreiro2008, Williams2017}. This is particularly well-known for the seminal Bell state measurement \cite{Lutkenhaus1999}, and it typically means that it is significantly harder to reach high total detection efficiencies with such measurements. Moreover, even implementing such measurements in dimensions larger than two is a formidable challenge. Using protocols based on product measurements offers a clear advantage. For instance, in the symmetric EAPM scenario implemented with qubits, we require a detection efficiency per photon of roughly $57.7\%$. Recent Bell inequality experiments have shown single photon detection efficiencies far above this value \cite{Liu2018, Shalm2021}. In contrast, we are not aware of any relevant two-photon Bell state measurement implemented with an efficiency close to this value. Notably, the theoretical  efficiency per photon needed in entanglement certification via dense coding is the same as in our protocol in the symmetric EAPM scenario. Furthermore, thanks to the dimensional scalability of product measurements, both Proposition~\ref{prop1} and Conjecture~\ref{Conj} suggests that the advantages in detection efficiency are even more significant for larger dimensions, as the efficiency threshold per photon will decrease monotonically with $d$. For instance, recent experiments on entangled four-dimensional photons show detection efficiencies around 71.7\% \cite{Hu2022}, well above the regime needed for protocols of our type.

\begin{acknowledgments}
C.R.C.~and A.T.~are supported  by the  Wenner-Gren Foundation, by the Knut and Alice Wallenberg Foundation through the Wallenberg Center for Quantum Technology (WACQT) and the Swedish Research Council under Contract No.~2023-03498. P. B. also acknowledges funding from the European Research Council (Consolidator grant `Cocoquest' 101043705). M.M. acknowledges funding from the DFG/FWF Research Unit FOR 2724 `Thermal machines in the quantum world'.
\end{acknowledgments}

\bibliography{references_manuscript}

\appendix

\section{Proof of $L_d$}\label{App:firstgame}
We prove that for $d$-dimensional states $\alpha_x$ and measurements $\{M_{c|z}\}$ it holds that
\begin{equation}\label{quantity}
\max_{\{\alpha_x\}, \{M_{c|z}\}} \frac{1}{d^2(d+1)}\sum_{x,z} \tr\left(\alpha_x M_{w_z|z}\right) \leq  \frac{1}{d}\left(1+\frac{d-1}{\sqrt{d+1}}\right)\equiv L_d.
\end{equation}
The proof closely parallels that  used to arrive at Result~1 in Ref.~\cite{farkas2023simple}. 

Trivially re-express the objective function on the left-hand-side of \eqref{quantity} as
\begin{multline}\label{reexpress}
\mathcal{S}_d=\frac{1}{d^2(d+1)}\sum_{x,z} \frac{1}{d}\tr\left(M_{w_z|z}\right)\\
+\frac{1}{d^2(d+1)}\sum_{x} \tr\left(\alpha_x \sum_z \left(M_{w_z|z}-\frac{\openone}{d}\tr(M_{w_z|z})\right)\right).
\end{multline}
Note that because of the winning conditions
\begin{align}\nonumber
&z\neq d:  &&w_z=x_1-2zx_0 \mod{d} \\
&z=d: && w_d=x_0,
\end{align}
it follows that  $\sum_x M_{w_z|z}=d\openone$ for every $z$. Therefore, the first term in \eqref{reexpress} becomes simply $\frac{1}{d}$. For the second term in \eqref{reexpress}, the optimal $\alpha_x$ corresponds to the largest eigenvalue of the operator $O_x=\sum_z \left(M_{w_z|z}-\frac{\openone}{d}\tr(M_{w_z|z})\right)$. Thus we have
\begin{equation}
\mathcal{S}_d=\frac{1}{d}+\frac{1}{d^2(d+1)}\sum_{x} \|O_x\|_\infty,
\end{equation}
where $\|\cdot \|_\infty$ is the largest modulus eigenvalue.  We then use norm inequality from Ref.~\cite{farkas2023simple}: for any trace-zero Hermitian $O$ it holds that
\begin{equation}
   \|O\|_\infty\leq \sqrt{\frac{\text{rank}(O)-1}{\text{rank}(O)}} \|O\|_F,
\end{equation}
where $\|O\|_F=\sqrt{\tr(OO^\dagger)}$ is the Frobenius norm. Applying this to each $O_x$ and then using the concavity of the square-root function yields
\begin{equation}\label{EqSS}
    \mathcal{S}_d\leq \frac{1}{d}+\frac{1}{d(d+1)}\sqrt{\frac{d-1}{d}} \sqrt{\sum_x \tr(O_x^2)}.
\end{equation}
We proceed with examining $I=\sum_x \tr(O_x^2)$. It becomes 
\begin{equation}\label{EqS}
I=\sum_x \sum_{z',z=1}^{d+1} \tr(M_{w_{z'}|z'}M_{w_z|z}) -\frac{1}{d}\sum_x \sum_{z',z=1}^{d+1} \tr(M_{w_{z'}|z'}) \tr(M_{w_z|z}).
\end{equation}
Label the first term $I_1$ and the second term $I_2$. We evaluate them one by one.
\begin{align}
I_1=\sum_x \sum_{z} \tr(M_{w_z|z}^2)+\sum_x \sum_{z'\neq z} \tr(M_{w_{z'}|z'}M_{w_z|z}).
\end{align}
Consider the second term. For a given pair $(z',z)$, we can define the set $T_{z',z}^c$ as the set of all pairs $(x_0,x_1)$ such that $w_{z'}=c$. (i) When  $z'=d$, characterising $T_{z',z}^c$ is trivial since its elements simply correspond to the pairs $\{(c=x_0,x_1)\}_{x_1}$. It is easily seen that $\{T_{z',z}^c\}_{c=0}^{d-1}$ is a parition of the set of all pairs $(x_0,x_1)$. (ii) When $z'\neq d$ and $z=d$, we let $T_{z',z}^c$ correspond to all pairs $(x_0,x_1)$ where $x_0=(2z')^{-1}(x_1-c)$. Note that the modular inverse always exists and is unique when $d$ is prime.  $\{T_{z',z}^c\}_{c=0}^{d-1}$ is a parition of the set of all pairs $(x_0,x_1)$. (iii) When $z'\neq d$ and $z\neq d$ the elements of $T_{z',z}^c$ correspond to choosing $x_0=(2z')^{-1}(x_1-c)$ for $x_1\in\{0,\ldots,d-1\}$. This  gives $w_z=x_1\left(1-2z(2z')^{-1}\right)+2z(2z')^{-1}c$. Again, $\{T_{z',z}^c\}_{c=0}^{d-1}$ is a parition of the set of all pairs $(x_0,x_1)$. 

Over all three cases, it holds that  for every $(z',z,c)$,
\begin{equation}\label{step}
\sum_{x\in T_{z',z}^c} M_{w_z|z}=\openone.
\end{equation}
Thus, we can write
\begin{align}\nonumber
&\sum_x \sum_{z'\neq z} \tr(M_{w_{z'}|z'}M_{w_z|z})=\sum_{z'\neq z}\sum_{c=0}^{d-1} \sum_{x\in T_{z',z}^c} \tr(M_{w_{z'}|z'}M_{w_z|z})\\ \nonumber
& =\sum_{z'\neq z}\sum_{c=0}^{d-1}  \tr\bigg(M_{c|z'}\sum_{x\in T_{z',z}^c} M_{w_z|z}\bigg)=\sum_{z'\neq z}\sum_{c=0}^{d-1}  \tr(M_{c|z'})\\ 
&=d \sum_{z'\neq z}=d^2(d+1).
\end{align}
In the first step, we used that $T_{z',z}^{c}$ is a partition of the set of $x$. In the second step, we use that $w_{z'}=c$ for every $x\in T_{z',z}^{c}$. In the third step, we use Eq.~\eqref{step} when $(z',z)\neq d$ and similarly when either $z'=d$ or $z=d$. Thus, we have
\begin{equation}
I_1=\sum_x \sum_{z} \tr(M_{w_z|z}^2)+d^2(d+1).
\end{equation}

Similarly, for the term $I_2$, we obtain the lower bound
\begin{align}
I_2\geq \frac{1}{d}\sum_x \sum_{z}\tr(M_{w_z|z}^2)+\frac{1}{d}\sum_x \sum_{z'\neq z}\tr(M_{w_{z'}|z'})\tr(M_{w_z|z}).
\end{align}
Using again the partition $\{T_{z',z}^c\}_c$ for the set of $x$, the second term reduces to becomes $d^2(d+1)$.

Inserting the above back into Eq.~\eqref{EqS}, we obtain
\begin{align}\nonumber
I&\leq  \frac{d-1}{d}\sum_x \sum_{z} \tr(M_{w_z|z}^2)\leq \frac{d-1}{d}\sum_x \sum_{z} \tr(M_{w_z|z})\\
&=d(d^2-1).
\end{align}
Inserting this into Eq.~\eqref{EqSS} we obtain the final result
\begin{align}
\mathcal{S}_d\leq \frac{1}{d}+\frac{d-1}{d\sqrt{d+1}}.
\end{align}

\section{Qubit balanced XOR schemes}\label{App:qubitproof}
We consider the case of qubit communication ($d=2$) in the symmetric EAPM scenario. We allow Alice, Bob and Charlie to have arbitrary inputs, with alphabet sizes $N_X$, $N_Y$ and $N_Z$ inputs each. Charlie's output is binary,  $c\in\{0,1\}$. In general, for each $z$, Charlie is tasked with outputing the value of some function of Alice's and Bob's inputs, $c=f_z(x,y)$. Its average success rate is
\begin{equation}
\mathcal{W}=\frac{1}{N_X N_YN_Z}\sum_{x,y,z}p(c=f_z(x,y)|x,y,z).
\label{eq: av success qubit}
\end{equation}
We focus on the broad class of schemes in which the winning condition is an XOR game with balanced functions, i.e.~any choice of $f_z(x,y)$ such that 
\begin{equation}
    f_z(x,y)=g_z(x)+ h_z(y) \mod{2},
\end{equation}
for some arbitrary functions $g_z$ and $h_z$ that are balanced. Recall that a function $u:\{1,\ldots,M\}\rightarrow \{0,1\}$ is called balanced if half the domain is mapped to $0$ and the other half is mapped to $1$. This means that $N_X$ and $N_Y$ must be even numbers.

We now derive an upper bound on $\mathcal{W}$ for any quantum model without shared entanglement. In such models, the average success rate reads
\begin{equation}
\mathcal{W}=\frac{1}{N_XN_YN_Z}\sum_{c,x,y,z}\tr\left(\alpha_x\otimes\beta_y M_{c|z}\right)\delta_{c,f_z(x,y)}.
\end{equation}
Using the normalisation $M_{0|z}+M_{1|z}=\openone$, we can express this as
\begin{equation}
\mathcal{W}=\frac{1}{2}+\frac{1}{N_XN_YN_Z}\sum_{x,y,z}(-1)^{f_z(x,y)}\tr\left(\alpha_x\otimes\beta_y M_{0|z}\right).
\end{equation}
Using that $f_z(x,y)=g_z(x)+ h_z(y)$, we re-arrange this as
\begin{equation}
	\mathcal{W}=\frac{1}{2}+\frac{1}{N_XN_YN_Z}\sum_{z}\tr\Bigg(M_{0|z}O^A_z\otimes O^B_z\Bigg).
 \label{eq: W Oz form}
\end{equation}
where
\begin{align}
&O^A_z=\sum_x (-1)^{g_z(x)}\alpha_x\\
&O^B_z=\sum_y (-1)^{h_z(y)}\beta_y.
\end{align}
The optimal choice of $M_{0|z}$ is the projector onto the positive eigenspace of the operator $O^A_z\otimes O^B_z$. Thus, the optimal value of the above trace is the sum of the positive eigenvalues of $O^A_z\otimes O^B_z$. Note that since $g_z$ and $h_z$ are balanced, $\tr(O^A_z)=\tr\left(O^B_z\right)=0$. Together with the fact that $O^A_z$ and $O^B_z$ are $2\times 2$ Hermitian operators, it follows that their two respective eigenvalues have the same magnitude and opposite sign, i.e.~$\lambda_1(O^A_z)=-\lambda_2(O^A_z)$ and $\lambda_1(O^B_z)=-\lambda_2(O^B_z)$, where $\lambda_1$ denotes the positive eigenvalue. Hence, the sum of positive eigenvalues of $O^A_z\otimes O^B_z$ becomes $\lambda_1(O^A_z)\lambda_1(O^B_z)+\lambda_2(O^A_z)\lambda_2(O^B_z)=2\lambda_1(O^A_z)\lambda_1(O^B_z)$. Hence,
\begin{equation}\nonumber
	\mathcal{W}\leq \frac{1}{2}+\frac{2}{N_XN_YN_Z}\sum_{z}\lambda_1(O^A_z)\lambda_1(O^B_z).
\end{equation}
We can now use the Bloch vector formalism to write $\alpha_x=(\openone+\vec{\alpha}_x\cdot \vec{\sigma})/2$ and $\beta_y=(\openone+\vec{\beta}_y\cdot \vec{\sigma})/2$, where $\vec{\sigma}=(\sigma_X,\sigma_Y,\sigma_Z)$, for some unit vectors $\{\vec{\alpha}_x\}$ and $\{\vec{\beta}_y\}$ in $\mathbb{R}^3$. %\moha{---Note that the choice of unit vector excludes mixed states which are sub-optimal due to convexity of the winning condition.} %\moha{update--actually, I am not sure how one can prove the convexity. Alice could send a mixed state in which the ``probabilities'' also depend on the parameters. how to prove convexity in this case?}
We can now write $O^A_z=\frac{1}{2}\vec{a}_z\cdot \vec{\sigma}$ and $O^B_z=\frac{1}{2}\vec{b}_z\cdot \vec{\sigma}$, where we have defined the unnormalised vectors $\vec{a}_z=\sum_{x}(-1)^{g_z(x)}\vec{\alpha}_x$ and $\vec{b}_z=\sum_{y}(-1)^{h_z(x)}\vec{\beta}_y$. It is easily shown that the  eigenvalues of an operator $\vec{n}\cdot\vec{\sigma}$ are $\pm |\vec{n}|$. Thus, we arrive at
\begin{align}\nonumber
	\mathcal{W}&\leq \frac{1}{2}+\frac{1}{2N_XN_YN_Z}\sum_{z} |\vec{a}_z||\vec{b}_z|\\
	&\leq \frac{1}{2}+\frac{1}{2N_XN_YN_Z}\sqrt{\sum_{z} |\vec{a}_z|^2}\sqrt{\sum_{z} |\vec{b}_z|^2},
\end{align}
where in the second line we have used the Cauchy-Schwarz inequality. This has the advantage that we can now consider the optimisation problem separately for each square-root factor. The expressions under the square-roots can be expanded to $\sum_z |\vec{a}_z|^2=N_XN_Z+2\eta$ and $\sum_z |\vec{b}_z|^2=N_YN_Z+2\xi$, where
\begin{align}
	&{\eta}\equiv \sum_{x<x'} \vec{\alpha}_x\cdot \vec{\alpha}_{x'} \sum_z(-1)^{g_z(x)+g_z(x')}\\
	& \xi\equiv \sum_{y<y'} \vec{\beta}_y\cdot \vec{\beta}_{y'}\sum_z  (-1)^{h_z(y)+h_z(y')}.
\end{align}
Thus, one is left with optimising these two expressions independently. One possible way of doing this is to define the Gram matrix $G_{x,x'}=\alpha_{x}\cdot \alpha_{x'}$.  Note that this matrix  by construction is both symmetric ($G=G^T$) and positive semidefinite ($G\succeq 0$). In addition, its diagonal elements are unit ($G_{x,x}=1$). Hence, we can relax our optimisation of $\eta$ to the semidefinite program 
\begin{equation}
\begin{aligned}\label{eq:alpha_tilde}
\tilde{\eta}=\max_G \quad & \sum_{x<x'} G_{x,x'}\left(\sum_z (-1)^{g_z(x)+g_z(x')}\right)\\
\textrm{s.t.}\quad & G\succeq 0,\qquad  G=G^T, \qquad G_{xx}=1 \qquad \forall \,x
\end{aligned}.
\end{equation}
%\moha{would it be still sdp if you add $-1 \leq G_{xx'}\leq 1$? maybe then it's a tight optimisation?\\}
Similarly, we can define a Gram matrix over the Bloch vectors $\{\vec{\beta}_y\}$ and bound $\xi$ by the analogous semidefinite program.
\begin{equation}
\begin{aligned}\label{eq:beta_tilde}
\tilde{\xi}=\max_G \quad & \sum_{y<y'} G_{y,y'}\left(\sum_z (-1)^{h_z(y)+h_z(y')}\right)\\
\textrm{s.t.}\quad & G\succeq 0,\qquad  G=G^T, \qquad G_{yy}=1 \qquad \forall \,y
\end{aligned}
\end{equation}
Thus, we can systematically compute bounds on the form
\begin{equation}\label{appbound}
\mathcal{W}\leq \frac{1}{2}+\frac{1}{2N_XN_YN_Z}\sqrt{N_XN_Z+2\tilde{\eta}}\sqrt{N_YN_Z +2\tilde{\xi}}.
\end{equation}

We now apply this to the specific scheme considered in the main text, i.e.~the quantity $\mathcal{R}_2$. Clearly, the symmetry between Alice and Bob means $\tilde{\eta}=\tilde{\xi}$. Moreover, for the above semidefinite program, all coefficients appearing in front of the Gram matrix in the objective function are negative one. Therefore, the program is invariant under permutations of the label $x$. Therefore, the semidefinite program simplifies to  
\begin{equation}\label{eq:alpha_tilde}
\begin{aligned}
\tilde{\eta}=&\max \quad -6u   \\
\textrm{s.t.}\quad &\begin{pmatrix}
1 & u & u&u \\
u & 1 & u &u\\
u&u&1 &u\\
u&u&u&1
\end{pmatrix}\succeq 0
\end{aligned}.
\end{equation}
As the distinct eigenvalues of the matrix are $1-u$ and $1+3u$, it follows that the optimal choice is $u=-\frac{1}{3}$ and thus $\tilde{\eta}=2$. Inserted into \eqref{appbound}, we obtain  $\mathcal{R}_2\leq\frac{2}{3}$.

\section{Connection with qubit maximally entangled fraction}\label{App:SFqubit}

Here, 
%we consider the case of qubit communication in the presence of a shared arbitrary \textit{pure} state in the symmetric EAPM scenario that we introduced in Section~\ref{sec3}. 
we would like to prove that Proposition 2 in Section~\ref{sec3} is tight for shared arbitrary \textit{pure} states. To this end, we find an upper bound that is equal to Eq.~\eqref{qubitSF}.
% average success rate given by Eq.\,\eqref{eq: av success qubit} in terms of maximally entangled fraction. 
%To do so, we first need to classify the shared pure states based on the maximally entangled fraction. 

Firstly, note that, any two-qubit pure state can be written in its Schmidt decomposition as
\begin{align}
\ket{\Psi_{AB}(\theta)}&= \cos{\theta}\,\vert{\psi_A, \psi_B\rangle}+\, \sin{\theta}\,\vert{\psi_{A}^{\perp}, \psi_{B}^{\perp}\rangle}\nonumber\\
&= V^A \otimes V^B \left(\cos{\theta}\,\vert{0_A, 0_B\rangle}+\, \sin{\theta}\,\vert{1_{A}, 1_{B}\rangle}\right)\nonumber\\
&= \, V^A \otimes V^B \,\vert{\Phi_{AB}(\theta)\rangle},
\label{eq:Schmidt dec}
\end{align}
where $\{\psi_{\gamma},\, \psi^{\perp}_{\gamma}\}$ represent an arbitrary orthogonal bases for the party $\gamma\in\{A,B\}$, $V^{\gamma}$ is a local unitary acting on the same Hilbert space, and $\theta \in [0,\pi/4] $. 
%can take values from be $0\leq \theta \leq \pi /4$. 
%nd $0\leq \gamma \leq \pi $.
One can show that the maximally entangled fraction of such a state is independent of the local unitaries, and only determined by $\theta$ %\cite{Brask2022} \armin{This is not the original citation. But I dont know what the original is}\parnam{I checked the cited paper but there is no citation.  }\armin{so let's just remove the citation}
\begin{equation}
\text{EF}_2\coloneqq\text{EF}_2(\Psi_{AB}(\theta))=\frac{1}{2}(1+\sin{2\theta}).
\label{eq: EF pure}
\end{equation}
%Eq.\,(\ref{eq:Schmidt dec}) also tells us how one can obtain all the pure states with a fixed maximally entangled fraction by local unitary operations.  So it provides a good opportunity to investigate the correlation game for pure state classified based on the singlet fraction.
% \begin{lemma}
% 	Any quantum strategy for the symmetric qubit EAPM  scenario that uses ``pure'' shared 
%  quantum states between Alice and Bob---fully characterised by Eq.~\eqref{eq:Schmidt dec}---is limited by
% 	\begin{equation}\label{qubitresult}
% 		\mathcal{R}_2\leq \frac{1}{3}+\, \frac{2}{3}\, \text{EF}_2(\theta).
% 	\end{equation}
% 	%This bound holds for local/global POVM measurements.
% \end{lemma}
%\begin{proof}
   %As we discussed all the pure states $\Psi_{AB}(\theta)=\ketbra{\Psi(\theta)}{\Psi_{AB}(\theta)}$ with the same singlet fraction can be described via real parameter $\theta $  and local unitaries. 
   
   Secondly, take the same  arbitrary state $\Psi_{AB}(\theta)$ given in Eq.~\eqref{eq:Schmidt dec}. 
   %For such states, we first try to find an upper-bound for the average success rate in balanced XOR scheme  the figure of merit can be written in the form of  
   Upon the application of local channels $\Lambda_{x}^{A}$ ($\Lambda_{y}^{B}$) by Alice (Bob), and the measurement $M_{c|z}$ by Charlie---not necessarily separable, the winning score reads %\armin{why do we have the Ns here? They are fixed to 4,3,4}\moha{what follows---until C14---is true for arbitrary Ns---i.e., the generalised games of App. B. If you prefer we can just restrict to 4,4,3, though. Let me know and I'll fix it quickly.}\armin{the result is only for the game we discuss in the main text. I dont see any reason to not specify the parameters}\moha{Please have a look now.}\armin{looks good}
\begin{align}
	%\mathcal{W}
 %&=\frac{1}{48}\sum_{x,y,z}\tr{U^A_{x_0x_1} \otimes U^B_{y_0y_1}\,(V^A \otimes V^B) \ketbra{\Phi_{AB}(\theta)}{\Phi_{AB}(\theta)}(V^A)^\dagger \otimes (V^B)^\dagger  \,(U^A_{x_0x_1})^\dagger\otimes (U^B_{y_0y_1})^\dagger M_{w_{xyz}|z}}\nonumber\\
 %&=\frac{1}{N_xN_yN_z}\sum_{x,y,z}\tr\bigg((\tilde{\Lambda}^A_{x} \otimes \tilde{\Lambda}^B_{y})\, [\Phi_{AB}(\theta)] \, M_{f_z(x,y)|z}\bigg).
 \mathcal{R}_2 & = \frac{1}{48}\sum_{x,y,z}\tr\bigg((\tilde{\Lambda}^A_{x} \otimes \tilde{\Lambda}^B_{y})\, [\Phi_{AB}(\theta)] \, M_{f_z(x,y)|z}\bigg).
\end{align}
where we also absorbed the local unitaries $V^{\gamma}$ into the local channels such that $\Tilde{\Lambda}^{\gamma}_{\alpha} [\bullet] = \, \Lambda_{\alpha}^\gamma \, [V^\gamma \bullet V^{\gamma{\dagger}}]$---since Alice and Bob have the chance to optimise their local operations anyways.  %$\Tilde{\Lambda}_{\alpha}^\gamma= \, \Lambda_{\alpha}^\gamma \, V^\gamma$ in which $\Lambda^\gamma$ is an arbitrary CPTP acting on the Hilbert space of system $\gamma$. 
%To obtain the bound we needed to optimise on all the local channels $ \Lambda_{\alpha}^\gamma$s and the pure states $V^{\gamma}$s.
%Since we optimize on all the local channels (CPTP) and the unitary channels are also the subset of local channels, we just need to optimize on all the channels.
%It means that the reachable upper bound for all the pure states with a fixed EF is the same.     
%We use the normalisation $M_{0|z}+M_{1|z}=\openone$ to express the correlations in terms of the zero-outcome operator, which we refer to simply as $M_z$. %Defining $O_z= \sum_{x,y} (-1)^{w_{xyz}}\alpha_x\otimes \beta_y$,
The average score rates can be further simplified to
\begin{equation}\label{c5}
	% \mathcal{W}
 % =\frac{1}{2}+\frac{1}{N_XN_YN_Z}\sum_{x,y,z}\tr\bigg((-1)^{f_z(x,y)}(\tilde{\Lambda}^A_{x} \otimes \tilde{\Lambda}^B_{y})\, [\Phi_{AB}(\theta)] \, M_{0|z}\bigg).
 	\mathcal{R}_2
 =\frac{1}{2}+\frac{1}{48}\sum_{x,y,z}\tr\bigg((-1)^{f_z(x,y)}(\tilde{\Lambda}^A_{x} \otimes \tilde{\Lambda}^B_{y})\, [\Phi_{AB}(\theta)] \, M_{0|z}\bigg).
\end{equation}
Next, note that the shared state $\Phi_{AB}(\theta)$ can be expressed in terms of Pauli matrices
\begin{align}
	\Phi_{AB}(\theta)=&\frac{\cos^2{\theta}}{4}(\openone+\sigma_3)\otimes (\openone+\sigma_3) \nonumber\\
 +&\,\frac{\sin^2{\theta} }{4}(\openone-\sigma_3)\otimes (\openone-\sigma_3)\nonumber\\
 +& \frac{\sin{2\theta}}{4}(\sigma_1\otimes \sigma_1\,-\, \sigma_2\otimes \sigma_2).
% =&\frac{1}{2}+\frac{1}{48\times 4}\sum_{x,y,z}\, tr[(-1)^{w_{xyz}}\tilde{U}^A_{x_0x_1} \otimes \tilde{U}^B_{y_0y_1}\, \big[ \,Z\otimes Z+ {\sin{2\theta}}(X\otimes X\,-\, Y\otimes Y)
 \label{eq: Pure state Pauli}
\end{align}
%Using that $f_z(x,y)=g_z(x)\oplus h_z(y)$ and \textcolor{red}{the following property which is true at least for the unital channel, i.e. $\Lambda[\openone]=\openone$}\moha{I am not sure about this. The channel can be a fixed channel for example, which is non-unital. Instead, can we somehow prove the best strategy for Alice/Bob is to use local unitaries---hence unital? Also prove that, on the other hand, the singlet fraction in C2 is obtained via a local unitary--rather than channels. The latter seems easier to proof.

Furthermore, since $\tilde{\Lambda}_{\alpha}$ is a quantum channel, for any arbitrary Bloch vector $\vec m$ with $|\vec m|\leq 1$ we should have  
\begin{align}
    \tilde{\Lambda}_\alpha [\,\frac{1}{2}(\openone +\vec{m} \cdot \vec{\sigma} )\,]& = \,  \frac{1}{2}(\openone + \vec{r}_{\alpha} \cdot \vec{\sigma} ) \label{eq: CPTP property 1}, \\
    \tilde{\Lambda}_\alpha [\,\vec{m}\cdot \vec{\sigma}\, ]&= \,\vec{s}_\alpha\cdot\vec{\sigma} ,
    \label{eq: CPTP property 2} 
\end{align}
where $|\vec r_{\alpha}|\leq 1$---since the right hand side of Eq.~\eqref{eq: CPTP property 1} is a normalised state. Moreover, we should also have $|\vec s_{\alpha}|\leq 1$. To see this, firstly applying the channel to the maximally mixed state yields $\tilde \Lambda_{\alpha} [\frac{\openone}{2}] = \frac{1}{2} (\openone + \vec r^0_{\alpha} \cdot \vec\sigma)$ for some $|\vec r^0_{\alpha}|\leq 1$. 
%If $|\vec r_{\alpha}^0|=0$ then the channel is unital, thus $\vec s_{\alpha} = \vec r_{\alpha}$ thus $|\vec s_{\alpha}|\leq 1$. For $0<|\vec r_{\alpha}^0| \leq 1 $ proceed as follows. 
Next, take two density matrices $\varrho_{\pm}= \frac{1}{2}(\openone \pm\vec{m} \cdot \vec{\sigma} )$, and apply the channel to them. The linearity of the channel dictates that
\begin{align*}
    \tilde \Lambda_{\alpha}[\varrho_{\pm}] = \tilde \Lambda_{\alpha}[\frac{\openone}{2}] \pm \frac{1}{2}\tilde \Lambda_{\alpha}[ \vec{m} \cdot \vec{\sigma}] = \frac{\openone}{2} + \frac{1}{2}(\vec r_{\alpha}^0 \pm \vec s_{\alpha})\cdot\vec\sigma.
\end{align*} 
Now, if $ |\vec s_{\alpha}| > 1$, at least one of the two vectors $\vec r_{\alpha}^0 \pm \vec s_{\alpha}$ has a norm above one, which contradicts the physicality of the channel. Thus we should always have $\vert \vec s_{\alpha} \vert \leq 1$. 
% to get note that, and  $\vert\vec{m} \vert \, \leq \, 1$, the norm of the vector $\vec{m}_\alpha $ must be smaller or equal to one. The first equation comes from the fact the channel $\tilde{\Lambda}_\alpha$ maps a two-dimensional state to another two-dimensional state. In Eq.\,\eqref{eq: CPTP property 2}, due to the trace-preserving property of channel, the identity matrix does not appear in RHS. Furthermore, for the norm condition, let's assume that for a vector $\vec{m}$ with $\vert\vec{m}\vert \, \leq \, 1$, there exists a channel $ \tilde{\Lambda}_\alpha$ such that $ \tilde{\Lambda}_\alpha [\,\vec{m}. \vec{\sigma}\, ]= \,\vec{m}_\alpha.\vec{\sigma}$ where $\vert\vec{m}_\alpha \vert \, \geq \, 1$ and $\tilde{\Lambda}_\alpha [\,\openone \,] = \, (\openone + \vec{n}_{\alpha} .\vec{\sigma} ) $ . Since $ \tilde{\Lambda}_\alpha$ is an CPTP channel must map both state $\varrho_{\pm}= \frac{1}{2}(\openone \pm\vec{m} . \vec{\sigma} )$ to valid two-dimensional states which are given by
% \begin{align}
%     \tilde{\Lambda}_\alpha [\,\frac{1}{2}(\openone \pm\vec{m} . \vec{\sigma} )\,] = \,  \frac{1}{2}(\openone + (\vec{n}_{\alpha}\pm\vec{m}_{\alpha}) .\vec{\sigma} ) ,
%     \label{eq: CPTP contradiction example} 
% \end{align}
% where the vectors $\vec{n}_{\alpha}$ and $\vec{m}_{\alpha}$ must satisfy the condition  $\vert \vec{n}_{\alpha}\pm\vec{m}_{\alpha}\vert \,\leq \, 1$ which is impossible to be true for $\vert\vec{m}_\alpha \vert \, \leq \, 1$. So  $\vert\vec{m}_\alpha \vert \, \leq \, 1$ is a necessary condition.

Applying the local channels to $\Phi_{AB}(\theta)$ gives
%Then, after acting the channel, the state can be written as 
\begin{align}
	(\tilde{\Lambda}^A_{x} \otimes \tilde{\Lambda}^B_{y}&)\, [\Phi_{AB}(\theta)] \,=\frac{\cos^2{\theta}}{4}(\openone+\,\vec{m}^{1}_{x}\cdot\vec{\sigma} )\otimes (\openone+\vec{n}^{1}_{y}\cdot\vec{\sigma} ) \nonumber\\
+&\,\frac{\sin^2{\theta} }{4}(\openone+\vec{m}^{2}_{x}\cdot\vec{\sigma})\otimes (\openone+\vec{n}^{2}_{y}\cdot\vec{\sigma})\nonumber\\
 +& \frac{\sin{2\theta}}{4}(\vec{m}_{x}^{3}\cdot\vec{\sigma}\otimes \vec{n}_{y}^{3}\cdot\vec{\sigma} \,+\, \vec{m}_{x}^{4}\cdot \vec{\sigma} \otimes \vec{n}_{y}^{4}\cdot\vec{\sigma} ),
% =&\frac{1}{2}+\frac{1}{48\times 4}\sum_{x,y,z}\, tr[(-1)^{w_{xyz}}\tilde{U}^A_{x_0x_1} \otimes \tilde{U}^B_{y_0y_1}\, \big[ \,Z\otimes Z+ {\sin{2\theta}}(X\otimes X\,-\, Y\otimes Y)
 \label{eq: Pure state Pauli after channel}
\end{align}
for some vectors $\{\vec{m}^{k}_{x}, \vec{n}^{k}_{y}\}_{k=1}^4$ that satisfy $|\vec{m}^{k}_{x}|\leq 1$ and $|\vec{n}^{k}_{y}|\leq 1$---note that we have absorbed a minus sign in the vector $\vec m_x^4$.
By plugging into Eq.\,\eqref{c5}, and noting that $f_z(x,y) = g_z(x) + h_z(y) \mod{2}$, with balanced functions $g_z(x)$ and $h_z(y)$, we have
\begin{align}
	% \mathcal{W}&
 % =\frac{1}{2}+\frac{1}{N_XN_YN_Z}\\
 % &\times \sum_{x,y,z} \nonumber\sum_{i=1}^{4}\tau_i \,\tr\big((-1)^{g_z(x)\oplus h_z(y)}\,  \vec{m}_{x}^i\cdot\vec{\sigma} \otimes\, \vec{n}_{y}^i\cdot \vec{\sigma} \, M_{0|z}\big),
 	\mathcal{R}_2&
 =\frac{1}{2}+\frac{1}{48}\\
 &\times \sum_{x,y,z} \nonumber\sum_{i=1}^{4}\tau_i \,\tr\big((-1)^{g_z(x) + h_z(y)}\,  \vec{m}_{x}^i\cdot\vec{\sigma} \otimes\, \vec{n}_{y}^i\cdot \vec{\sigma} \, M_{0|z}\big),
\end{align}
where $  \tau= \frac{1}{4}(\cos^2(\theta),\sin^2(\theta),\sin{2\theta},\,\sin{2\theta})\geq 0$ within the domain of $\theta$. 
%and $\vec{\alpha}_{3,\gamma}=\, \cos^{2}{\theta} \,\vec{\alpha}^{+}_{3,\gamma}+ \sin^{2}{\theta} \,\vec{\alpha}^{-}_{3,\gamma}  $ for $\alpha\in\,\{m,\,n\}$ and $\gamma\in\,\{x,\,y\}$. 
We have also made use of the fact that for any balanced function $f(x)$, $\sum_{x}(-1)^{f(x)}=0 .$ 

We can bring $\mathcal{R}_2$ to a similar form as Eq.\,\eqref{eq: W Oz form},   
\begin{equation}
	% \mathcal{W}=\frac{1}{2}+\frac{4}{N_XN_YN_Z}\sum_{i=1}^{4} \tau_i \,\sum_{z}\tr\Bigg(M_{0|z}O^A_{z,i}\otimes O^B_{z,i}\Bigg).
 \mathcal{R}_2=\frac{1}{2}+\frac{1}{12}\sum_{i=1}^{4} \tau_i \,\sum_{z}\tr\Bigg(M_{0|z}O^A_{z,i}\otimes O^B_{z,i}\Bigg).
 \label{eq: W Oz form pure state}
\end{equation}
where
\begin{align}
&O^A_{z,i}=\frac{1}{2}\sum_x (-1)^{g_z(x)}\vec{m}_{x}^i\cdot\vec{\sigma},\\
&O^B_{z,i}=\frac{1}{2}\sum_y (-1)^{h_z(y)}\vec{n}_{y}^i\cdot\vec{\sigma}.
\end{align}
One can separately bound each of the four terms by using the result obtained in Appendix \ref{App:qubitproof}---see Eqs.~\eqref{appbound} and \eqref{eq:alpha_tilde}---to get 
\begin{equation}
	% \sum_{z}\tr\big(M_{0|z}O^A_{z,i}\otimes O^B_{z,i}\big) \, \leq \, \frac{\sqrt{N_XN_Z+2\tilde{\alpha}}\sqrt{N_YN_Z +2\tilde{\beta}}\hspace{0.2cm}}{2}\forall i,
 \sum_{z}\tr\big(M_{0|z}O^A_{z,i}\otimes O^B_{z,i}\big) \, \leq \, 8,
 %\forall i,
 \label{eq: W Oz form pure state}
\end{equation}
%where $\tilde{\alpha}$ and $\tilde{\beta}$ are defined in Eqs.~\eqref{eq:alpha_tilde} and \eqref{eq:beta_tilde} and are independent of ``$i$''.
which is independent of ``$i$''. Therefore,
% \begin{equation}
% 	\mathcal{W}\leq\,\frac{1}{2}+\frac{2}{N_XN_YN_Z}\sum_{i=1}^{4} \tau_i \,\sqrt{N_XN_Z+2\tilde{\alpha}}\sqrt{N_YN_Z +2\tilde{\beta}}.
%  \label{eq: W gen bound pure state}
% \end{equation}

% Focusing on the symmetric EAPM scenario, where we have $N_X=N_Y=4$, $N_Z=3$, and ${\tilde \alpha} = {\tilde\beta} = 2$---see Eq.~\eqref{eq:alpha_tilde}, the bound reduces to
\begin{align}
	\mathcal{R}_2\, \leq \,\frac{1}{2}+\frac{2}{3}\sum_{i=1}^{4} \tau_i\, =\frac{1}{3}+\, \frac{2}{3}\, \text{EF}_2(\theta). 
 \label{eq: W upperbound}
\end{align}
where we have used Eq.~\eqref{eq: EF pure} to evaluate $\sum_{i=1}^{4} \tau_i=\, \text{EF}_2(\theta)-\frac{1}{4}$.

%\end{proof}
\section{Connection with maximally entangled fraction for arbitrary $d$}\label{App:SFqudit}
First, we prove that $\mathcal{R}_d=1$. For simplicity, we define the projector $E_{m|z}=\vert e_{m,z}\rangle \langle e_{m,z}\vert$. We measure one particle in the basis $\{E_{c_1|z}\}$, with outcome $c_1$, and the other particle in the conjugated basis $\{E^*_{c_2|z}\}$, with outcome $c_2$. The final output is defined as $c=c_1-c_2 \mod{d}$. Thus, the POVM describing this measurement becomes
	\begin{equation}
	M_{c|z}=\sum_{c_1,c_2} E_{c_1|z}\otimes E_{c_2|z}^* \delta_{c,c_1-c_2}.
	\end{equation}
	Next, via a direct calculation, it is possible to show these relations (given in the main text) for any $z=0,\ldots,d-1$,
	\begin{align}\nonumber \label{lemma}
	& X^t \vert{e_{m,z}\rangle}=\omega^{zt^2-tm} \vert{e_{m-2zt,z}\rangle}\\\nonumber
	& X^t \vert{e^*_{m,z}\rangle}=\omega^{-zt^2+tm} \vert{e^*_{m-2zt,z}\rangle}\\\nonumber
	& Z^t \vert{e_{m,z}\rangle}= \vert{e_{m+t,z}\rangle}\\
	& Z^t \vert{e^*_{m,z}\rangle}=\vert{e^*_{m-t,z}\rangle}.
	\end{align}
	That is, applying any unitary of the form $X^{t_1}Z^{t_2}$ to any of the eigenstates associated to the first $d$ bases are mapped into other eigenstates of the same basis. 
	
	Consider now the probability to output the right answer, $c=w_z$, for $z=0,\ldots,d-1$. It reads
	\begin{multline}
	p(c=w_z|x,y,z)=\\
	\sum_{c_1}\bracket{\phi^+_d}{\left(\vert \nu_{c_1xz}\rangle\langle \nu_{c_1xz}\vert \otimes \vert \mu_{c_1xyz}\rangle\langle \mu_{c_1xyz}\vert \right)}{\phi^+_d}
	\end{multline}
	where 
	\begin{align}
	& \vert{\nu_{c_1xz}\rangle}=Z^{-x_1}X^{-x_0} \vert{e_{c_1,z}\rangle} \\
    &\vert{\mu_{c_1xyz}\rangle}=Z^{-y_1}X^{-y_0} \vert{e^*_{c_1-w_z,z}\rangle}
    \end{align}
	Successively applying the relations \eqref{lemma}, one finds that
	\begin{align}\nonumber\label{q4}
	& 
 \vert \nu_{c_1xz}\rangle \langle \nu_{c_1xz}\vert =E_{c_1+2zx_0-x_1|z}\\
	& 
 \vert \mu_{c_1xyz}\rangle \langle \mu_{c_1xyz}\vert=E^*_{c_1+2zx_0-x_1|z}.
	\end{align}
	Notice that the second vector now has no dependence on $y$. From the relation, $\openone \otimes O \vert{\phi^+_d\rangle}=O^T \otimes \openone \vert{\phi^+_d\rangle}$ and the fact that $\{E_{m|z}\}$ forms a orthonormal basis, it follows that
	\begin{multline}
	p(w_z|x,y,z)=\frac{1}{d}\sum_{c_1}\Tr\left(E_{c_1+2zx_0-x_1|z}E_{c_1+2zx_0-x_1|z}\right)=1.
	\end{multline}
	Performing the same calculation for the computational basis, $z=d$, gives an analogous result. Hence, we have perfect correlations for every $z$ and hence achieve the algebraically optimal value $\mathcal{R}_d=1$.

The maximally entangled fraction is defined as
\begin{equation}
\text{EF}_d(\rho)=\max_{\Lambda_1,\Lambda_2} \bracket{\phi^+}{(\Lambda_1 \otimes \Lambda_2) [\rho]}{\phi^+_d}
\end{equation}
where $\Lambda_1$ and $\Lambda_2$ are CPTP maps with $d$-dimensional output spaces. If it exceeds $1/d$, the state is entangled. To show that this quantity is relevant for the value of $\mathcal{R}_d$, let Alice and Bob first apply local CPTP maps  $\Lambda_1$ and $\Lambda_2$, respectively, before applying the protocol used above to arrive at $\mathcal{R}_d=1$ for the maximally entangled state. 

We can write the figure of merit as
	\begin{align}
	\mathcal{R}_d=\frac{1}{d+1}\tr\left( (\Lambda_1 \otimes \Lambda_2) [\rho] \sum_{z=0}^{d} \mathcal{R}^{(z)}_d\right),
	\end{align}
	where we have defined
	\begin{align}
	\mathcal{R}^{(z)}_d= \frac{1}{d^4}\sum_{x,y,c_1}
 \vert \nu_{c_1xz} \rangle \langle \nu_{c_1xz}\vert \otimes \vert \mu_{c_1xyz} \rangle \langle \mu_{c_1xyz}\vert
\end{align}
	for $z=0,\ldots,d-1$.  Using \eqref{q4} and taking the sums over $(x,y,c_1)$, we get
	\begin{equation}
	\mathcal{R}^{(z)}_d= \sum_{c=0}^{d-1} E_{c|z}\otimes E_{c|z}^*.
	\end{equation}
	This can be thought of as an unnormalised correlated-coin state in the $z$'th MUB. Indeed, a direct calculation for the computational basis, namely $z=d$, analogously leads to $\mathcal{R}^{(d)}_d=\sum_{c=0}^{d-1}\ketbra{cc}$. Now, we can use the key fact that for any complete set of mutually unbiased bases, it holds that \cite{Morelli2023}
	\begin{equation}
	\mathcal{T}\equiv \sum_{z=0}^{d} \sum_{c=0}^{d-1}E_{c|z}\otimes E_{c|z}^*=\openone + d \phi^+_d.
	\end{equation}
	Hence, we have
	\begin{align}\nonumber
	\mathcal{R}_d&=\frac{1}{d+1}\tr\left( (\Lambda_1 \otimes \Lambda_2) [\rho] \mathcal{T}\right)\\
	&=\frac{1}{d+1}+\frac{d}{d+1}\tr\left( (\Lambda_1 \otimes \Lambda_2) [\rho]\phi^+_d\right).
	\end{align}
	Since we can allow any channels for Alice and Bob, we can choose those that correspond to the maximally entangled fraction of $\rho$. Hence, we have obtained
		\begin{equation}
	\mathcal{R}_d=\frac{1}{d+1}+\frac{d}{d+1}\text{EF}_d(\rho).
	\end{equation}

\end{document}